\documentclass[
submission
]{dmtcs-episciences}

\usepackage[utf8]{inputenc}
\usepackage{subfigure}
\usepackage{amsmath}
\usepackage{amssymb}
\usepackage{amscd}

\newtheorem{definition}{Definition}[section]

\newtheorem{theorem}[definition]{Theorem}

\author{Michael Haythorpe\affiliationmark{1}}
\title[Reducing the generalised Sudoku problem to the Hamiltonian cycle problem]{Reducing the generalised Sudoku problem to the Hamiltonian cycle problem}
\affiliation{
  Flinders University, Australia}
\received{2016-03-09}
\revised{xxxx-xx-xx}
\accepted{yyyy-yy-yy}
\begin{document}
\publicationdetails{VOL}{2015}{ISS}{NUM}{SUBM}
\maketitle
\begin{abstract}
  The generalised Sudoku problem with $N$ symbols is known to be NP-complete, and hence is equivalent to any other NP-complete problem, even for the standard restricted version where $N$ is a perfect square. In particular, generalised Sudoku is equivalent to the, classical, Hamiltonian cycle problem. A constructive algorithm is given that reduces generalised Sudoku to the Hamiltonian cycle problem, where the resultant instance of Hamiltonian cycle problem is sparse, and has $O(N^3)$ vertices. The Hamiltonian cycle problem instance so constructed is a directed graph, and so a (known) conversion to undirected Hamiltonian cycle problem is also provided so that it can be submitted to the best heuristics. A simple algorithm for obtaining the valid Sudoku solution from the Hamiltonian cycle is provided. Techniques to reduce the size of the resultant graph are also discussed.
\end{abstract}
\section{Introduction}\label{sec-Introduction}

The {\em generalised Sudoku problem} is an NP-complete problem which, effectively, requests a Latin square that satisfies some additional constraints. In addition to the standard requirement that each row and column of the Latin square contains each symbol precisely once, Sudoku also demands {\em block constraints}. If there are $N$ symbols, the Latin square is of size $N \times N$. If $N$ is a perfect square, then the Latin square can be divided into $N$ regions of size $\sqrt{N} \times \sqrt{N}$, called {\em blocks}. Then the block constraints demand that each of these blocks also contain each of the symbols precisely once. Typically, the symbols in a Sudoku puzzle are simply taken as the natural numbers $1$ to $N$. In addition, Sudoku puzzles typically have fixed values in some of the cells, which dramatically limits the number of valid solutions. If the fixed values are such that only a unique solution remains, the Sudoku puzzle is said to be {\em well-formed}.

The standard version where $N = 9$ has, in recent years, become a common form of puzzle found in newspapers and magazines the world over. Although variants of the problem have existed for over a century, Sudoku in its current format is a fairly recent problem, first published in 1979 under the name Number Place. The name Sudoku only came into existence in the 1980s. In 2003, the generalised Sudoku problem was shown to be ASP-complete \cite{asp}, which in turn implies that it is NP-complete. Hence, it is theoretically as difficult as any problems in the set $\mathcal{NP}$ of decision problems for which a positive solution can be certified in polynomial time. Note that although there are more general versionsvariants of Sudoku (such as rectangular versions), the square variant described above where $N$ is a perfect square suffices for NP-completeness. Hence, for the remainder of this manuscript, it will be assumed that we are restricted to considering the square variant.

Since being shown to be NP-complete, Sudoku has subsequently been converted to various NP-complete problems, most notably constraint satisfaction \cite{csp}, boolean satisfiability \cite{sat} and integer programming \cite{ip}. Another famous NP-complete problem is the {\em Hamiltonian cycle problem} (HCP), which is defined as follows. For a simple graph (that is, one containing no self-loops or multi-edges) containing vertex set $V$ and edge set $E : V \rightarrow V$, determine whether any simple cycles containing all vertices in $V$ exist in the graph. Such cycles are called {\em Hamiltonian cycles}, and a graph containing at least one Hamiltonian cycle is called {\em Hamiltonian}. Although HCP is defined for directed graphs, in practice most heuristics that actually solve HCP are written for undirected graphs.

Since both Sudoku and HCP are NP-complete, it should be possible to reduce Sudoku to HCP. In this manuscript, a constructive algorithm that constitutes such a reduction is given. The resultant instance of HCP is a sparse graph or order $O(N^3)$. If many values are fixed, it is likely that the resultant graph can be made smaller by clever graph reduction heuristics; to this end, we apply a basic graph reduction heuristic to two example Sudoku instances to investigate the improvement offered.

It should be noted that reductions of NP-complete problems to HCP is an interesting but still largely unexplored field of research. Being one of the classical NP-complete problems (indeed, one of the initial 21 NP-complete problems described by Karp \cite{karp}), HCP is widely studied and several very efficient algorithms for solving HCP exist. HCP is also an attractive target problem in many cases because the resultant size of the instance is relatively small by comparison to other potential target problems. Indeed, the study of which NP-complete problems provide the best target frameworks for reductions is an ongoing field of research. For more on this topic, as well as examples of other reductions to HCP, the interested reader is referred to \cite{dewdney,creignou,setsplitting,hcp23hcp}.

\section{Conversion to HCP}\label{sec-conversion}

At it's core, a Sudoku problem with $N$ symbols (which we will consider to be the natural numbers from 1 to $N$) has three sets of constraints to be simultaneously satisfied.

\begin{enumerate}\item Each of the $N$ blocks must contain each number from 1 to $N$ precisely once.
\item Each of the $N$ rows must contain each number from 1 to $N$ precisely once.
\item Each of the $N$ columns must contain each number from 1 to $N$ precisely once.\end{enumerate}

The variables of the problem are the $N^2$ cells, which can each be assigned any of the $N$ possible values, although some of the cells may have fixed values depending on the instance.

In order to cast an instance of Sudoku as an instance of Hamiltonian cycle problem, we need to first encode every possible variable choice as a subgraph. The idea will be that traversing the various subgraphs in certain ways will correspond to particular choices for each of the variables. Then, we will link the various subgraphs together in such a way that they can only be consecutively traversed if none of the constraints are violated by the variable choices.

In the final instance of HCP that is produced, the vertex set $V$ will comprise of the following, where $a$, $i$, $j$ and $k$ all take values from $1$ to $N$:

\begin{itemize}\item A single starting vertex $s$ and finishing vertex $f$
\item Block vertices: $N^2$ vertices $b_{ak}$, corresponding to number $k$ in block $a$
\item Row vertices: $N^2$ vertices $r_{ik}$, corresponding to number $k$ in row $i$
\item End Row vertices: $N$ vertices $t_i$ corresponding to row $i$
\item Column vertices: $N^2$ vertices $c_{jk}$ corresponding to number $k$ in column $j$
\item End Column vertices: $N$ vertices $d_j$ corresponding to column $j$
\item Puzzle vertices: $3N^3$ vertices $x_{ijkl}$ corresponding to number $k$ in position $(i,j)$, for $l = 1, 2, 3$
\item End Puzzle vertices: $N^2$ vertices $v_{ij}$ corresponding to position $(i,j)$
\item Duplicate Puzzle vertices: $3N^3$ vertices $y_{ijkl}$ corresponding to number $k$ in position $(i,j)$, for $l = 1, 2, 3$
\item End Duplicate Puzzle vertices: $N^2$ vertices $w_{ij}$ corresponding to position $(i,j)$\end{itemize}

The graph will be linked together in such a way that any valid solution to the Sudoku puzzle will correspond to a Hamiltonian cycle in the following manner.

\begin{enumerate}\item The starting vertex $s$ is visited first.
\item For each $a$ and $k$, suppose number $k$ is placed in position $(i,j)$ in block $a$. Then, vertex $b_{ak}$ is visited, followed by all $x_{ijml}$ for $m \neq k$, followed by all $y_{ijml}$ for $m \neq k$. This process will ensure constraint 1 is satisfied.
\item For each $i$ and $k$, suppose number $k$ is placed in position $(i,j)$ in row $i$. Then, vertex $r_{ik}$ is visited, followed by $x_{ijk3}$, $x_{ijk2}$, $x_{ijk1}$ and then $v_{ij}$. If $k = N$ (ie if $i$ is about to be incremented or we are finished step 3) then this is followed by $t_i$. This process will ensure constraint 2 is satisfied.
\item For each $j$ and $k$, suppose number $k$ is placed in position $(i,j)$ in column $j$. Then, vertex $c_{jk}$ is visited, followed by $y_{ijk3}$, $y_{ijk2}$, $y_{ijk1}$ and then $w_{ij}$. If $k = N$ (ie if $j$ is about to be incremented or we are finished step 4) then this is followed by $d_j$. This process will ensure constraint 3 is satisfied.
\item The finishing vertex $f$ is visited last and the Hamiltonian cycle returns to $s$.\end{enumerate}

What follows is a short description of how steps 1--5 are intended to work. A more detailed description follows in the next section.

The idea of the above is that we effectively create two identical copies of the Sudoku puzzle. In step 2, we place numbers in the puzzles, which are linked together in such a way to ensure the numbers are placed identically in both copies. Placing a number $k$ into position $(i,j)$, contained in block $a$, is achieved by first visiting $b_{ak}$, and then proceeding to visit every puzzle vertex $x_{ijml}$ {\bf except} for when $m = k$, effectively leaving the assigned number \lq\lq open", or unvisited. Immediately after visiting the appropriate puzzle vertices, the exact same duplicate puzzle vertices $y_{ijml}$ are visited as well, leaving the assigned number unvisited in the second copy as well. Since each block vertex $b_{ak}$ is only visited once, each number is placed precisely once in each block, satisfying constraint 1. The hope is, after satisfying constraint 1, that the row and column constraints have also been satisfied. If not, it will prove impossible to complete steps 3 and 4 without needing to revisit a vertex that was visited in step 2.

In step 3, we traverse the row vertices one at a time. If number $k$ was placed in position $(i,j)$, then row vertex $r_{ik}$ is followed by the unvisited vertices $x_{ijk3}$, $x_{ijk2}$, $x_{ijk1}$, and then by the end puzle vertex $v_{ij}$. Once all $r_{ik}$ vertices have been traversed for a given $i$, we visit the end row vertex $t_i$. Note that the three $x$ vertices visited for each $i$ and $k$ in step 3 are the three that were skipped in step 2. Therefore, every puzzle vertex is visited by the time we finish traversing all the row vertices. However, if row $i$ is missing the number $k$, then there will be no available unvisited puzzle vertices to visit after $r_{ik}$, so this part of the graph can only be traversed if all the row constraints are satisfied by the choices in step 2.

Step 4 evolves analogously to step 3, except for $c_{jk}$ instead of $r_{ik}$, $y_{ijkl}$ instead of $x_{ijkl}$, $w_{ij}$ instead of $v_{ij}$ and $d_j$ instead of $t_i$. Hence, this part of the graph can only be traversed if all the column constraints are also satisfied by the choices in step 2.

Assuming the graph must be traversed as described above, it is clear that all Hamiltonian cycles in the resultant instance of HCP correspond to valid Sudoku solutions. In order to show this is the case, we first describe the set of directed edges $E$ in the graph. Note that in each of the following, if $k+1$ or $k+2$ are bigger than $N$, they should be wrapped back around to a number between $1$ and $N$ by subtracting $N$. For example, if $k+2 = N+1$ then it should be taken as 1 instead.

\begin{itemize}\item $(s\;,\;b_{11}), (d_N\;,\; f)$ and $(f\;,\;s)$
\item $(b_{ak}\;,\; x_{i,j,(k+1),1})$ for all $a, k$, and $(i,j)$ contained in block $a$
\item $(x_{ijk1}\;,\;x_{ijk2}), (x_{ijk2}\;,\;x_{ijk1}), (x_{ijk2}\;,\;x_{ijk3})$ and $(x_{ijk3}\;,\;x_{ijk2})$ for all $i, j, k$
\item $(x_{ijk3}\;,\;x_{i,j,(k+1),1})$ for all $i, j, k$
\item $(y_{ijk1}\;,\;y_{ijk2}), (y_{ijk2}\;,\;y_{ijk1}), (y_{ijk2}\;,\;y_{ijk3})$ and $(y_{ijk3}\;,\;y_{ijk2})$ for all $i, j, k$
\item $(y_{ijk3}\;,\;y_{i,j,(k+1),1})$ for all $i, j, k$
\item $(x_{ijk3}\;,\;y_{i,j,(k+2),1})$ for all $i, j, k$
\item $(y_{ijk3}\;,\;b_{a,k+2})$ for all $i, j$, and for $k \neq N-1$, where $a$ is the block containing position $(i,j)$
\item $(y_{i,j,N-1,3}\;,\;b_{a+1,1})$ for all $i, j$ except for the case where both $i = N$ and $j = N$, where $a$ is the block containing position $(i,j)$
\item $(y_{N,N,N-1,3}\;,\;r_{11})$
\item $(r_{ik}\;,\;x_{ijk3})$ for all $i, j, k$
\item $(x_{ijk1}\;,\;v_{ij})$ for all $i, j, k$
\item $(v_{ij}\;,\;r_{ik})$ for all $i, j, k$
\item $(v_{ij}\;,\;t_i)$ for all $i, j$
\item $(t_i\;,\;r_{i+1,1})$ for all $i < N$
\item $(t_N\;,\;c_{11})$
\item $(c_{jk}\;,\;y_{ijk3})$ for all $i, j, k$
\item $(y_{ijk1}\;,\;w_{ij})$ for all $i, j, k$
\item $(w_{ij}\;,\;c_{jk}$ for all $i, j, k$
\item $(w_{ij}\;,\;d_j)$ for all $i, j$
\item $(d_j\;,\;c_{j+1,1})$ for all $j < N$\end{itemize}

\section{Detailed explanation}\label{sec-explanation}

We need to show that every valid Hamiltonian cycle corresponds to a valid Sudoku solution. Note that at this stage, we have not handled any fixed cells, so any valid Sudoku solution will suffice. Fixed cells will be taken care of in Section \ref{sec-fixed}.

\begin{theorem}Every Hamiltonian cycle in the graph constructed in the previous section corresponds to a valid Sudoku solution, and every valid Sudoku solution has corresponding Hamiltonian cycles.\end{theorem}

\begin{proof}First of all, note that vertices $x_{ijk2}$ are degree 2 vertices, and so they ensure that if vertex $x_{ijk1}$ is visited before $x_{ijk3}$, it must be proceeded by $x_{ijk2}$ and then $x_{ijk3}$. Likewise, if vertex $x_{ijk3}$ is visited before $x_{ijk1}$, it must be proceeded by $x_{ijk2}$ and $x_{ijk1}$. The same argument holds for vertices $y_{ijk2}$. This will ensure that the path any Hamiltonian cycle must take through the $x$ and $y$ vertices is tightly controlled.

Each of the block vertices $b_{ak}$ links to $x_{i,j,(k+1),1}$ for all $(i,j)$ contained in block $a$. One of these edges must be chosen. Suppose number $k$ is to be placed in position $(i,j)$, contained in block $a$. Then the edge $(b_{ak},x_{i,j,(k+1),1})$ is traversed. From here, the cycle must continue through vertices $x_{i,j,(k+1),2}$ and $x_{i,j,(k+1),3}$. It is then able to either exit to one of the $y$ vertices, or continue visiting $x$ vertices. However, as will be seen later, if it exits to the $y$ vertices at this stage, it will be impossible to complete the Hamiltonian cycle. So instead it continues on to $x_{i,j,(k+2),1}$, and so on. Only once all of the $x_{ijml}$ vertices for $m \neq k$ have been visited (noting that $i$ and $j$ are fixed here) can it safely exit to the $y$ vertices -- refer this as Assumption 1 (we will investigate later what happens if Assumption 1 is violated for any $i,j,k$). The exit to $y$ vertices will occur immediately after visiting vertex $x_{i,j,(k-1),3}$, which is linked to vertex $y_{i,j,(k+1),1}$. Note that by Assumption 1, vertices $x_{ijkl}$ are unvisited for $l = 1, 2, 3$. Then, from the $y$ vertices, the same argument as above applies again, and eventually vertex $y_{i,j,(k-1),3}$ is departed, linking to vertex $b_{a,k+1}$ if $k < N$, or to vertex $b_{a+1,1}$ if $k = N$. Refer to the equivalent assumption on visiting the $y$ vertices as Assumption 2. This continues until all the block vertices have been traversed, at which time vertex $y_{N,N,N-1,3}$ links to $r_{11}$. Note that, other than by violating Assumptions 1 or 2, it is not possible to have deviated from the above path. By the time we arrive at $r_{11}$, all the block vertices $b_{ak}$ have been visited. Also, every puzzle vertex $x_{ijkl}$ and duplicate puzzle vertex $y_{ijkl}$ has been visited other than those corresponding to placing number $k$ in position $(i,j)$.

Next, each of the row vertices $r_{ik}$ links to $x_{ijk3}$ for all $i, j, k$. For each $i$ and $k$, one of these edges must be chosen. However, by Assumption 1, all vertices $x_{ijk3}$ have already been visited except for those corresponding to the number $k$ being placed in position $(i,j)$. If the choices in the previous step violate the row constraints, then there will be a row $i$ that does not contain a number $k$, and subsequently there will be no valid edge emanating from vertex $r_{ik}$. Hence, if the choices made in step 2 violate the row constraints, and Assumption 1 is correct, it is impossible to complete a Hamiltonian cycle. If the choices in the previous step satisfy the row constraints, then there should always be precisely one valid edge to choose here. Once vertex $x_{ijk3}$ is visited, vertices $x_{ijk2}$ and $x_{ijk1}$ must follow, at which point the only remaining valid choice is to proceed to vertex $v_{ij}$. From here, any row vertex $r_{im}$ that has not yet been visited can be visited. If all, have been visited, then $t_i$ can be visited instead. Note that once $t_i$ is visited, it is impossible to return to any $r_{ik}$ vertices, so they must all be visited before $t_i$ is visited.

An analogous argument to above can be made for the column vertices $c_{jk}$. Note that if Assumptions 1 and 2 are correct, then vertex $y_{ijkl}$ will be unvisited at the start of step 4 if and only if $x_{ijkl}$ was unvisited at the start of step 3. Therefore, we see that if Assumptions 1 and 2 are correct, then it is only possible to complete the Hamiltonian cycle if the choices made in step 2 correspond to a valid Sudoku solution.

Now consider the situation where Assumption 1 is violated, that is, after step 2 there exists unvisited vertices $x_{ijkl}$ and $x_{ijml}$ for some $i, j$, and $k \neq m$. Then during step 3, without loss of generality, suppose vertex $r_{ik}$ is visited before $r_{im}$. As argued above, this will be followed by vertices $x_{ijk3}$, $x_{ijk2}$, $x_{ijk1}$, at which point visiting vertex $v_{ij}$ is the only available choice. Then later, $r_{im}$ is visited. It must visit $x_{ijm3}$, $x_{ijm2}$, $x_{ijm1}$ and is then, again, forced to proceed to vertex $v_{ij}$. However, since vertex $v_{ij}$ has already been visited, this is impossible and the Hamiltonian cycle cannot be completed. If Assumption 2 is violated, and it is vertices $y_{ijkl}$ and $y_{ijml}$ that are unvisited after step 2, an analogous argument can be made involving step 4. Hence, every Hamiltonian cycle in the graph must satisfy Assumptions 1 and 2. This completes the proof.\end{proof}

Since any valid Sudoku solution has corresponding Hamiltonian cycles, the resulting instance of HCP is equivalent to a blank Sudoku puzzle. In a later section, the method for removing edges based on fixed numbers for a given Sudoku instance is described. Since the instance of HCP can be constructed, and the relevant edges removed, in polynomial time as a function of $N$, the algorithm above constitutes a reduction of Sudoku to the Hamiltonian cycle problem.

\section{Size of \lq\lq blank" instance}\label{sec-blank}

The instance of HCP that emerges from the above conversion consists of $6N^3 + 5N^2 + 2N + 2$ vertices, and $19N^3 + 2N^2 + 2N + 2$ directed edges. For the standard Sudoku puzzle where $N = 9$, this corresponds to a directed graph with $4799$ vertices and $14033$ directed edges.

All of the best HCP heuristic currently available assume that the instance is undirected. There is a well-known conversion of directed HCP to undirected HCP which can be performed as follows. First, produce a new graph which has three times as many vertices as the directed graph. Then add edges to this new graph by the following scheme, where $n$ is the number of vertices in the directed graph:

\begin{enumerate}\item Add edges $(3i-1,3i-2)$ and $(3i-1,3i)$ for all $i = 1, \hdots, n$.
\item For each directed edge $(i,j)$ in the original graph, add edge $(3i,3j-2)$.\end{enumerate}

In the present case, this results in an undirected instance of HCP consisting of $18N^3 + 15N^2 + 6N + 6$ vertices and $31N^3 + 12N^2 + 6N + 6$ edges. This implies that the average degree in the graph grows monotonically with $N$, but towards a limit of $\frac{31}{9}$, so the resultant graph instance is sparse. For $N = 4$, the average degree is just slightly above $3.1$, and for $N = 9$ the average degree is just under $3.3$.

A trick can be employed to reduce the number of vertices in the undirected graph. Consider the vertices in the undirected graph corresponding to the $x$ and $y$ vertices. In particular, consider the set of 9 vertices corresponding to $x_{ijk1}$, $x_{ijk2}$ and $x_{ijk3}$. The nine vertices form an induced subgraph such as that displayed at the top of Figure \ref{fig-subgraph}. There are incoming edges incident on the first and seventh vertices, and outgoing edges incident on the third and ninth vertices. If the induced subgraph is entered via the first vertex, it must be departed via the ninth vertex, or else a Hamiltonian cycle cannot be completed. Likewise, if the induced subgraph is entered via the seventh vertex, it must be departed via the third vertex. It can be seen by inspecting all cases that if the fifth vertex is removed, and a new edge is introduced between the fourth and sixth vertices, the induced subgraph retains these same properties. This alternative choice is displayed at the bottom of Figure \ref{fig-subgraph}. Such a replacement can be made for each triplet $x_{ijkl}$ or $y_{ijkl}$. Hence, we can remove $2N^3$ vertices and $2N^3$ edges from the undirected graph for a final total of $16N^3 + 15N^2 + 6N + 6$ vertices and $29N^3 + 12N^2 + 6N + 6$, although at the cost of raising the average degree by a small amount (roughly between 0.1 and 0.15, depending on $N$.)

\begin{figure}[h!]\begin{center}\includegraphics[scale=0.35]{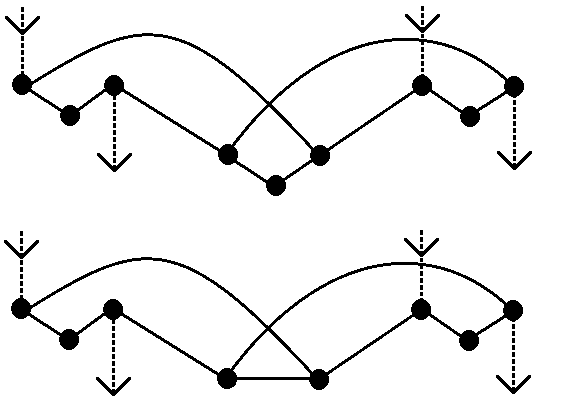}\caption{The induced subgraph created after the conversion to an undirected graph, corresponding to vertices $x_{ijk1}$, $x_{ijk2}$ and $x_{ijk3}$, and an alternative subgraph with one vertex removed.\label{fig-subgraph}}\end{center}\end{figure}

\section{Handling fixed numbers}\label{sec-fixed}

In reality, all meaningful instances of Sudoku have fixed values in some of the $N^2$ cells. Although this could potentially be handled by removing vertices, it would then be necessary to redirect edges appropriately. Instead, it is simpler to remove edges that cannot be used while choosing these fixed values. Once this is performed, a graph simplifying heuristic could then be employed to remove unnecessary vertices if desired.

For each fixed value, $12N - 12$ edges can be identified as redundant, and be removed. However, when there are multiple fixed values, some edges may be identified as redundant multiple times, so $12N - 12$ is only an upper bound on the number of edges that can be removed per fixed value. For example, suppose one cell has a fixed value of 1, and another cell within the same block has a fixed value of 2. From the first fixed value, we know that all other entries in the block must not be 1. From the second fixed value, we know that the second cell must have a value of 2, and hence not 1. Then the edge corresponding to placing a value of 1 in the second cell would be identified as redundant twice. The exact number of redundant edges identified depends on the precise orientation of the fixed values.

For each fixed value $k$ in position $(i,j)$, and block $a$ containing position $(i,j)$, the following sets of edges are redundant and may be removed (an explanation for each set follows the list):

\begin{itemize}\item[(1)] $(b_{ak}\;,\;x_{mnk1})$ for all choices of $m$ and $n$ such that block $a$ contains $(m,n)$, and also $(m,n) \neq (i,j)$
\item[(2)] $(b_{am}\;,\;x_{ijm1})$ for $m \neq k$
\item[(3)] $(x_{m,n,(k-1),3}\;,\;y_{m,n,(k+1),1})$ for all choices of $m$ and $n$ such that block $a$ contains $(m,n)$, and also $(m,n) \neq (i,j)$
\item[(4)] $(x_{i,j,(m-1),3}\;,\;y_{i,j,(m+1),1})$ for $m \neq k$
\item[(5a)] If $k < N$ : $(y_{m,n,(k-1),3}\;,\;b_{a,k+1})$ for all choices of $m$ and $n$ such that block $a$ contains $(m,n)$, and also $(m,n) \neq (i,j)$
\item[(5b)] If $k = N$ and $a < N$ : $(y_{m,n,(k-1),3}\;,\;b_{a+1,1})$ for all choices of $m$ and $n$ such that block $a$ contains $(m,n)$, and also $(m,n) \neq (i,j)$
\item[(5c)] If $k = N$ and $a = N$ : $(y_{m,n,(k-1),3}\;,\;r_{11})$ for all choices of $m$ and $n$ such that block $a$ contains $(m,n)$, and also $(m,n) \neq (i,j)$
\item[(6a)] $(y_{i,j,(m-1),3}\;,\;b_{a,m+1})$ for $m \neq k$ and $m \neq N$
\item[(6b)] If $k < N$ and $a < N$ : $(y_{i,j,(N-1),3}\;,\;b_{a+1,1})$
\item[(6c)] If $k < N$ and $a = N$ : $(y_{i,j,(N-1),3}\;,\;r_{11})$
\item[(7)] $(r_{ik}\;,\;x_{imk3})$ for all $m \neq k$
\item[(8)] $(x_{imk1}\;,\;v_{im})$ for all $m \neq k$
\item[(9)] $(r_{im}\;,\;x_{ijm3})$ for all $m \neq k$
\item[(10)] $(c_{jk}\;,\;y_{mjk3})$ for all $m \neq k$
\item[(11)] $(y_{mjk1}\;,\;w_{mk})$ for all $m \neq k$
\item[(12)] $(c_{jm}\;,\;y_{ijm3})$ for all $m \neq k$\end{itemize}

The edges in set (1) correspond to the option of placing a value of $k$ elsewhere in block $a$. The edges in set (2) correspond to the option of picking a value other than $k$ in position $(i,j)$. Those two sets of incorrect choices would lead to the edges from sets (3) and (4) respectively being used to transfer from the $x$ vertices to the $y$ vertices, and so those edges are also redundant.

The edges in (5a)--(5c) correspond to the edges that return from the $y$ vertices to the next block vertex after an incorrect choice is made (corresponding to the set (1)). If $k = N$ then the next block vertex is actually for the following block, rather than for the next number in the same block. If $k = N$ and $a = N$ then all block vertices have been visited and the next vertex is actually the first row vertex.

Likewise, the edges in (6a)--(6c) correspond to the edges that return from the $y$ vertices after an incorrect choice is made (corresponding to the set (2)). Note that if $k = N$, there are $N-1$ redundant edges in (6a). If $k < N$ there are $N-2$ redundant edges in (6a) and then one additional redundant edge from either (6b) or (6c).

The edges in set (7) correspond to the option of finding a value of $k$ in row $i$ at a position other than $(i,j)$, which is impossible. The edges in set (8) correspond to visiting the end puzzle vertex after making an incorrect choice from (7). The edges in set (9) correspond to the option of finding a value other than $k$ in row $i$ and position $(i,j)$, which is also impossible. Analogous arguments can be made for the edges in sets (10)--(12), except for columns instead of rows.

Each of sets (1)-(4) and (7)-(12) identify $N-1$ redundant edges each. As argued above, the relevant sets from (5a)--(5c) will contribute $N-1$ more redundant edges, as well the relevant sets from (6a)--(6c). Hence, the maximum number of edges that can be removed per number is $12N - 12$ for each fixed value.

\section{Recovering the Sudoku solution from a Hamiltonian cycle}\label{sec-solution}

The constructive algorithm above produces a HCP instance for which each solution corresponds to a valid Sudoku solution Once such a solution is obtained, the following algorithm reconstructs the corresponding Sudoku solution:

Denote by $h$ the Hamiltonian cycle obtained. For each $i = 1, \hdots, N$ and $j = 1, \hdots, N$, find vertex $v_{ij}$ in $h$. Precisely one of its adjacent vertices in $h$ will be of the form $x_{ijk1}$ for some value of $k$. Then, number $k$ can be placed in the cell in the $i$th row and $j$th column in the Sudoku solution.

Suppose that the vertices are labelled in the order given in Section \ref{sec-conversion}. That is, $s$ is labelled as 1, $f$ is labelled as 2, the $b_{ak}$ vertices are labelled $3, 4, \hdots, N^2-2$, and so on. Then, for each $i$ and $j$, vertex $v_{ij}$ will be labelled $3N^3 + 3N^2 + (i+1)N + (j+2)$, and vertex $x_{ijk1}$ will be labelled $3iN^2 + (3j-1)N + 3k$. Of course, if the graph has been converted to an undirected instance, or if it has been reduced in size by a graph reduction heuristic, these labels will need to be adjusted appropriately.

\section{Reducing the size of the HCP instances}\label{sec-reducing}

After constructing the HCP instances using the above method, graph reduction techniques can be applied. Most meaningful instances of Sudoku will have many fixed values, which in turn leads to an abundance of degree 2 vertices.

In order to test the effectiveness of such techniques, a very simple reduction algorithm was used. Iteratively, the algorithm iteratively checks the following two conditions until there are no applicable reductions remaining:

\begin{enumerate}\item If two adjacent vertices are both degree 2, they can be contracted to a single vertex.
\item If a vertex has two degree 2 neighbours, all of its incident edges going to other vertices can be removed.\end{enumerate}

Note that the second condition above leads to three adjacent degree 2 vertices which will in turn be contracted to a single vertex. The removal of edges when the second condition is satisfied often leads to additional degree 2 vertices being formed which allows the algorithm to continue reducing.

Note also that this simple graph reduction heuristic is actually hampered by the graph reduction method described in Section \ref{sec-blank}, since that method eliminates many degree 2 vertices. It is likely that a more sophisticated graph reduction heuristic could be developed that incorporates both methods.

The above heuristic was applied to both a well-formed (that is, uniquely solvable) Sudoku instance with 35 fixed values, as well as one of the Sudoku instances from the repository of roughly 50000 instances maintained by Royle \cite{royle}. The instances in that repository all contain precisely 17 fixed numbers, and are all well-formed; it was recently proved via a clever exhaustive computer search that 17 is the minimal number of fixed values for a well-formed Sudoku problem with 9 symbols \cite{min17}. The two instances tested are displayed in Figure \ref{fig-instances}.

\begin{figure}[h!]\begin{center}\includegraphics[scale=0.6]{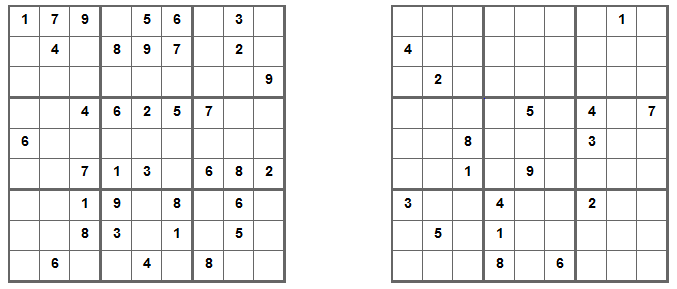}\caption{Two well-formed Sudoku instances with 35 fixed values and 17 fixed values respectively.\label{fig-instances}}\end{center}\end{figure}

After the simple reduction heuristic above was applied to the first Sudoku instance, it had been reduced from an undirected instance with 14397 vertices and 22217 edges, to an equivalent instance with 8901 vertices and 14175 edges. Applying the above reduction algorithm to the second Sudoku instance from Royle's repository reduced it from an undirected instance with 14397 vertices and 22873 edges, to an equivalent instance with 12036 vertices and 19301 edges. In both cases the reduction is significant, although obviously more there are greater opportunities for reduction when there are more fixed values.

Both instances were solved by Concorde \cite{concorde} which is arguably the best algorithm for solving HCP instances containing large amount of structure, as its branch-and-cut method is very effective at identifying sets of arcs that must be fixed all at once, or not at all, particularly in sparse graphs. Technically, Concorde actually converts the HCP instance to an equivalent TSP instance but does so in an efficient way. The first instance was solved during Concorde's presolve phase, while the second instance required 20 iterations of Concorde's branch and cut algorithm\footnote{It should be noted that Concorde does use a small amount of randomness in its execution. The random seed used in this experiment was 1453347272.} to discover a solution. This would seem to indicate that the first Sudoku instance can be solved without requiring any amount of guessing. The two solutions were then interpreted via the algorithm in Section \ref{sec-solution} to provide solutions to the initial Sudoku instances; those solutions are displayed in Figure \ref{fig-solutions}.

\begin{figure}[h!]\begin{center}\includegraphics[scale=0.6]{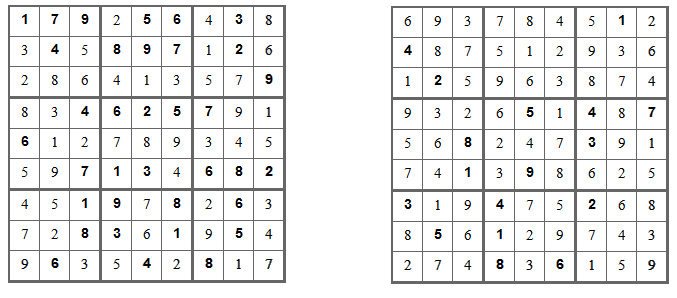}\caption{The solutions to the Sudoku instances in Figure \ref{fig-instances}, as interpreted from the Hamiltonian cycles of the converted HCP instances.\label{fig-solutions}}\end{center}\end{figure}

\end{document}